\documentclass{article}

\usepackage[utf8]{inputenc}

\title{On the geometry of the kinematic space in special relativity}
\author{Rafael Ferreira \and João dos Reis Junior \and Carlos H.~Grossi}
\date{}

\usepackage{graphicx}
\usepackage{bbm}
\usepackage{latexsym}
\usepackage{fullpage}
\usepackage[a4paper, top=3cm, bottom=3cm,left=2cm,right=2cm]{geometry}
\usepackage{times}
\usepackage{amsfonts, amsmath, amsthm, amssymb}
\usepackage{wrapfig}
\usepackage{indentfirst}

\newtheoremstyle{thmstyle}
{\topsep}
{\topsep}
{\itshape}
{0pt}
{\bfseries}
{.}
{5pt plus 1pt minus 1pt}
{#2.\hspace{3pt}#1#3}

\newtheoremstyle{defistyle}
{\topsep}
{\topsep}
{}
{0pt}
{\bfseries}
{.}
{5pt plus 1pt minus 1pt}
{#2.\hspace{3pt}#1#3}

\numberwithin{equation}{subsection}

\theoremstyle{thmstyle}

\newtheorem{lemma}[equation]{Lemma}
\newtheorem{prop}[equation]{Proposition}
\newtheorem{cor}[equation]{Corollary}
\theoremstyle{defistyle}
\newtheorem{defi}[equation]{Definition}
\newtheorem{rmk}[equation]{Remark}
\theoremstyle{defistyle}
\newtheorem{example}[equation]{Example}
\DeclareMathOperator{\arccosh}{arccosh}

\DeclareMathOperator{\ta}{ta}
\DeclareMathOperator{\T}{T}
\DeclareMathOperator{\G}{G}

\begin{document}

\maketitle

\begin{abstract}
The classifying space of inertial reference frames in special relativity is naturally hyperbolic. There is a remarkable interplay between central elements of hyperbolic geometry and those of special relativity --- which, to a certain extent, have already been observed in the past --- that we present and further discuss in the paper. We aim at a geometrization of special relativity at the level of kinematic space by giving to physical concepts/phenomena purely geometric definitions/descriptions. In this way, the differences between special relativity and classical mechanics can be seen as a manifestation of the distinct geometric natures of their kinematic spaces.
\end{abstract}

\section{Introduction}

A major conceptual difference between Newtonian mechanics and special relativity is that the kinematic space $\mathcal K$ of the first is Euclidean\footnote{It would be more accurate to say that it is just a vector space (with no distinguished metric), see \cite{pereiragrossi2021}.} while that of the former is hyperbolic, a fact already observed by V.~Vari\'cak in 1910 \cite{Var1910} and E.~Borel in 1913 \cite{borel2013_1}, \cite{borel2013_2}. Here, {\it kinematic space\/} is to be understood as the classifying space of all inertial reference frames (see Subsection \ref{sub:kinematicspace}).

The hyperbolic nature of special relativity has been explored by several authors from distinct perspectives. Some are based on the role played by {\it rapidity,} introduced by Vari\'cak and called {\it true velocity\/} by E.~Borel. Rapidity appears naturally in the context of special relativity because it is simply the hyperbolic distance between inertial reference frames, that is, it is the hyperbolic distance in $\mathcal K$. Another hyperbolic view on special relativity involves the use of gyrovector spaces, introduced by A.~Ungar (see, for instance, \cite{ungar2001}), which constitute an algebraic framework for hyperbolic geometry that builds upon an axiomatization of the (noncommutative and nonassociative) relativistic velocity addition.

The path we take in this paper focuses on some simple geometric invariants related to finite configurations of points in kinematic space. (It comes mainly from \cite{AGr2011}, where a coordinate-free toolbox that suits several ``classic'' geometries~--- including, for instance, hyperbolic, spherical, Fubini-Study, de Sitter, and anti de Sitter geometries --- is developed.) A~first example of such a geometric invariant is the \textit{tance\/} (see \eqref{eq:tance} for the definition) which is, in a certain sense, the simplest algebraic invariant of a pair of points in $\mathcal K$. The square root of the tance is a fundamental quantity in hyperbolic geometry because distance is a monotonic function of it. Curiously, when translated into the context of special relativity, the square root of the tance between two inertial observers in $\mathcal K$ is simply the Lorentz factor related to the observers (see Remark \ref{rmk:velocity}). Keeping up with this idea of translating into special relativity some natural concepts and geometric invariants in hyperbolic geometry, we obtain the following:
\begin{itemize}
\item The relative velocity between inertial observers $\pmb p,\pmb q\in\mathcal K$ appears as a natural algebraic expression for the tangent vector to the geodesic segment joining $\pmb p,\pmb q$ (see Definition \ref{defi:relative_velocity});
\item Rapidity and the closely related concept of scaled rapidity are shown to have distinct geometric origins; while rapidity measures the hyperbolic distance between inertial reference frames, scaled rapidity measures the hyperbolic distance between relative velocities (see Section \ref{sub:rapidityandvelocity});
\item Parallel transport gives rise to the relativistic velocity addition in a straightforward generalization of the classical velocity addition (see Definition \ref{defi:sum_of_rapidities});
\item Hypercycles (that is, curves equidistant from a geodesic in $\mathcal K$) allow one to write a ``parallelogram law'' for the relativistic velocity addition (see the end of Subsection \ref{sub:rapidityandvelocity});
\item The general relativistic Doppler effect can be described by a natural expression involving the Busemann function related to a photon or, equivalently, to a point in the ideal boundary of $\mathcal K$ (see Proposition \ref{prop:Doppler}); moreover, horospheres appear as level surfaces of energy/frequency (see Corollary \ref{cor:level_curves}). There is a striking resemblance between such geometric form of the relativistic Doppler effect and the study of probability measures in the context of Patterson-Sullivan theory (see \cite[Section 1.2 and Proposition 3.9]{quint2006} for the Patterson-Sullivan perspective);
\item A basic algebraic invariant involving two inertial observers in $\mathcal K$ and a pair of space-like separated events determines whether the observers agree or disagree on the order of occurrence of the events (see Subsection \ref{sub:causality});
\item Curves in $\mathcal K$ can be seen as describing the inertial reference frames occupied by an observer at each instant of its proper time and a tangent vector to such a curve gives the instantaneous $4$-acceleration of the observer. Hence, dynamics can also be modelled at the level of the kinematic space (see Subsection \ref{sub:dynamics}).
\end{itemize}

We arrive at what seems to be an effective geometrization of special relativity: physical concepts and phenomena (like the Lorentz factor, velocity, velocity addition, the Doppler effect, among others) gain a purely geometric description which does not depend on their actual definitions in physics. Moreover, the techniques that are used in the paper directly extended to Grassmannians \cite{AGr2012}, \cite{ananingoncalvesgrossi2019} and this allows one to deal in a similar fashion with special relativity in other Einstein geometries like anti de Sitter and de Sitter spacetimes.

It is worthwhile mentioning that, in our construction, kinematic space is naturally compactified by the de Sitter space as they are are glued along their common ideal boundaries. The interplay between these geometries, which are linked by the geometry of Minkowski space, is very rich. For instance, in the case of $4$-dimensional Minkowski space, there is a duality between points in the de Sitter component (which correspond to the sometimes called tachyonic inertial reference frames) and circles in the ideal boundary (which correspond to families of photons whose velocities, as measured by certain inertial observers, are all coplanar), see Remark \ref{rmk:tachyon}.

In spite of the emphasis we give on the geometric point of view, the synthetic and coordinate-free methods that we use provide simple explicit formulae for all the involved concepts (say, geodesics, parallel transport, Riemannian connection, curvature tensor, among others \cite{AGr2011}). These methods are essentially ``linear'' and they are also applicable to several other geometries which are common in physics; in this regard, see Subsection \ref{sub:classicgeometries} and Example~\ref{example:classicgeometries}.

Finally, developing a similar approach to classical mechanics requires one to take as spacetime a vector space equipped with a \emph{degenerate} symmetric bilinear form (of signature $0++\dots+$) in place of Minkowski space \cite{pereiragrossi2021}. In a certain sense, special relativity and classical mechanics arise from their kinematic spaces in the same way; however, being very different from each other, the geometric natures of such kinematic spaces give rise to completely distinct phenomenologies.

\setcounter{equation}{0}

\section{Preliminaries}\label{sec:preliminaries}
\subsection{Kinematic space}\label{sub:kinematicspace}
Let $\mathbb M^{n+1}$ be Minkowski $(n+1)$-space, that is, an $\mathbb R$-vector space equipped with a symmetric bilinear form $\langle-,-\rangle:\mathbb M^{n+1}\to\mathbb R$ of signature $-+\dots+$. As usual, the \textit{light cone\/} consists of the \textit{lightlike\/} vectors $v\in\mathbb M^{n+1}$ which satisfy $v\ne0$ and $\langle v,v\rangle=0$. Minkowski space is divided by the light cone into \textit{timelike\/} and \textit{spacelike\/} vectors, respectively characterized by $\langle v,v\rangle<0$ and $\langle v,v\rangle>0$. We also assume that one of the light cone sheets is chosen as the future light cone.

The $1$-dimensional subspace $\mathbb Rv\subset\mathbb M^{n+1}$, where $v$ is a timelike vector, can be seen as the worldline of an inertial reference frame. The space of all such worldlines consists of an open subspace of the real projective space $\mathbb P^n_\mathbb R$ and, topologically, this subspace is an open $n$-ball called the (open) \textit{kinematic space\/} $\mathcal K$. The boundary $\partial\mathcal K$ of $\mathcal K$ is an $(n-1)$-sphere consisting of the projectivization of the light cone; in other words, each point in $\partial\mathcal K$, an \textit{isotropic\/} point, represents the worldline of a photon. We call $\overline{\mathcal K}:=\mathcal K\cup\partial\mathcal K$ the \textit{closed\/} kinematic space and the entire projective space, the \textit{extended\/} kinematic space. Moreover, we denote by $\mathcal G$ the complement $\mathbb P_\mathbb R^n\setminus\overline{\mathcal K}$.

A point in projective space will be denoted by a bold letter and a representative of this point in Minkowski space, by the same roman letter; so, $\pmb p\in\mathbb P^n _\mathbb R$ stands for the equivalence class $\mathbb Rp$ of a point $p\in\mathbb M^{n+1}$. Strictly speaking, the points in kinematic space represent the worldlines of inertial observers that synchronised their clocks at a same point in spacetime (the vertex of the lightcone which corresponds to coordinate time $t=0$ for every inertial observer). By choosing a representative $p\in\mathbb M^{n+1}$ of a point $\pmb p\in\mathcal K$, we therefore pick a specific coordinate time $t=\pm|p|/c$ in the frame of the corresponding inertial observer ($c$ denotes the speed of light in vacuum). However, we will typically abuse nomenclature and refer to a point in $\mathcal K$ simply as an inertial observer (or inertial reference frame).

\begin{rmk}\label{rmk:three_frames}
When dealing with $3$ inertial reference frames or, equivalently, with three points in $\mathcal K$ (a configuration that will be considered several times in the paper), we can assume that $n=2$ because the vector space generated by these frames (equipped with the induced form) is precisely $\mathbb M^3$. In this case, the extended kinematic space is the real projective plane $\mathbb P_\mathbb R^2$ and the worldlines corresponding to photons give rise to a topological circle $\mathbb S^1$ which divides $\mathbb P^2 _\mathbb R$ into the open disk $\mathcal K$ and the open M\"obius band $\mathcal G=\mathbb P_\mathbb R^2\setminus\overline{\mathcal K}$.
\end{rmk}

{\bf Tangent space and metric.} The symmetric bilinear form in $\mathbb M^{n+1}$ canonically induces a Riemannian metric in the open kinematic space $\mathcal K$ as well as a Lorentzian metric in $\mathcal G$. Indeed, there is a natural identification
\begin{equation}\label{eq:tangentspace}
\mathrm{T}_{\pmb p}\mathbb P_\mathbb R^n=\mathrm{Lin}(\mathbb Rp,p^\perp)
\end{equation}
\noindent
between the tangent space to $\mathbb P_\mathbb R^n$ at a nonisotropic point $\pmb p\in\mathbb P_\mathbb R^n$ and the space of linear maps from $\mathbb Rp$ to its orthogonal complement $p^\perp$ with respect to the symmetric bilinear form. This identification may be interpreted in the following way (for a formal proof see, for instance, \cite[Subsection A.1.1]{AGG2011}). A tangent vector $\varphi\in\mathrm{T}_{\pmb p}\mathbb P_\mathbb R^n$ can be seen as representing a movement in its direction. When the point $\pmb p$ starts moving in the direction of $\varphi$, the corresponding subspace $\mathbb Rp$ rotates around the origin of $\mathbb M^{n+1}$ and such a rotation can be described in terms of a linear map $\mathbb Rp\to p^\perp$ as in Figure 1.

\begin{wrapfigure}{L}{0pt}\label{fig:tangent_vector}
\includegraphics[width=.2\textwidth]{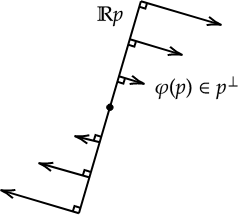}
\caption{Tangent vector}
\end{wrapfigure}

In view of the identification \eqref{eq:tangentspace}, given tangent vectors $\varphi_1,\varphi_2\in\mathrm{Lin}(\mathbb Rp,p^\perp)$ at a non-isotropic point $\pmb p\in\mathbb P_\mathbb R^n$, we define
\begin{equation}\label{eq:metric}
\langle \varphi_1,\varphi_2\rangle_{\pmb p}:=-\frac{\big\langle \varphi_1(p),\varphi_2(p)\big\rangle}{\langle p,p \rangle}. \end{equation}
This provides a semi-Riemannian metric in extended kinematic space outside isotropic points (note that the above formula does not depend on the choice of the representative for $\pmb p$).
This metric is actually Riemannian in the open kinematic space $\mathcal K$ because, in this case, the symmetric bilinear form, restricted to $p^\perp$, is positive-definite. It is called the \textit{hyperbolic\/} metric and endows $\mathcal K$ with a geometric structure equivalent to Klein's model of the hyperbolic $n$-ball. One can similarly see that \eqref{eq:metric} is a Lorentzian metric in $\mathcal G$, called the \textit{de Sitter\/} metric. The extended kinematic space is therefore the gluing, along isotropic points, of the kinematic space with the de Sitter space. In order to explore the interplay between $\mathcal K$ and $\mathcal G$, we need to introduce (extended) geodesics.

\smallskip

{\bf Extended geodesics and duality.}
An \textit{extended geodesic\/} is a projective line, that is, the projectivization $\mathbb P_\mathbb RW$ of a $2$-dimensional linear subspace $W\subset\mathbb M^{n+1}$. In particular, there exists a unique extended geodesic, denoted by $\mathrm{G}{\wr}\pmb p,\pmb q{\wr}$, that contains a pair of distinct points $\pmb p,\pmb q\in\mathbb P_\mathbb R^n$. Topologically, an extended geodesic is always a circle. The intersection of $\mathbb P_\mathbb RW$ with $\mathcal K$ (respectively, with $\mathcal G$) is, if non-empty, a usual geodesic in hyperbolic space (respectively, in de Sitter space). Moreover, all the geodesics in hyperbolic space, as well as in de Sitter space, appear in this way \cite{AGr2011}. The possible signatures of the symmetric bilinear form restricted to $W$ are $-+$, $+0$, and $++$. The first case provides all the geodesics in $\mathcal K$ and it is easy to see that each such geodesic has a pair of isotropic points, called its \textit{vertices.} In the case of de Sitter space, all the admissible signatures of $W$ appear: when $W$ is respectively of signatures $-+$, $0+$, or $++$, the corresponding geodesics have spacelike, lightlike, or timelike tangent vectors with respect to the Lorentzian metric \eqref{eq:metric}. Moreover, a geodesic has a pair of distinct isotropic vertices in the first case, a single isotropic vertex in the second case, and no isotropic points in the last case.

We can now see, by means of a simple duality, that the de Sitter space is nothing but the space of all geodesics in kinematic space when $n=2$. Indeed, given a point $\pmb p\in\mathcal G$, we obtain the geodesic $\mathbb P_\mathbb Rp^\perp\cap\mathcal K$ due to $p^\perp$ being of signature $-+$. The point $\pmb p$ is called the {\it polar point\/} of the geodesic $\mathbb P_\mathbb Rp^\perp\cap\mathcal K$. Conversely, given a geodesic $\mathbb P_\mathbb RW\cap\mathcal K$, we obtain the point $\mathbb P_\mathbb RW^\perp\in\mathcal G$. (Clearly, the kinematic space itself can be seen as the space of all timelike geodesics in $\mathcal G$ and the extended kinematic space, as the space of all geodesics in $\mathcal G$.) For arbitrary $n$, the de Sitter space is the space of all totally geodesic \textit{hyperplanes\/} in the kinematic space (a totally geodesic hyperplane in $\mathcal K$ is given by $\mathbb P_\mathbb RW\cap\mathcal K$ when $W$ is a codimension $1$ linear subspace of $\mathbb M^{n+1}$ of signature $-+\dots+$).

\begin{rmk}\label{rmk:tachyon}
Let $n=3$. Given $\pmb p\in\mathcal G$, the totally geodesic plane $P:=\mathbb P_\mathbb Rp^\perp\cap\mathcal K$ intersects the ideal boundary $\partial\mathcal K$ in a circle $C$. It follows from Definition \ref{defi:relative_velocity}, Proposition \ref{prop:velocity_is_velocity}, and from the fact that $P$ is totally geodesic that  any observer in $P$ agrees that the velocities of the photons corresponding to the points in the circle $C$ are coplanar. In other words, under the mentioned duality, one can see an inertial ``reference frame'' corresponding to a point in $\mathcal G$ (sometimes called a tachyonic worldline) as being equivalent to such a family of photons.
\end{rmk}

{\bf Tance.}
The length of the geodesic segment joining two inertial reference frames $\pmb p,\pmb q\in\mathcal K$ is the hyperbolic distance $d(\pmb p,\pmb q)$ between $\pmb p$ and $\pmb q$. It is given by $d(\pmb p,\pmb q)=\arccosh\sqrt{\ta(\pmb p,\pmb q)}$, where \begin{equation}\label{eq:tance}
\ta(\pmb p,\pmb q):=\frac{\langle p,q\rangle\langle q,p\rangle}{\langle p,p\rangle\langle q,q\rangle}
\end{equation}
is the \textit{tance\/} between $\pmb p,\pmb q$ \cite{AGr2011}. (In the next subsection, we will also refer to the tance in the case of a non-degenerate Hermitian form in a complex vector space; this is why we write its definition in this way.) The hyperbolic distance, also known in the context of special relativity as \textit{rapidity,} is therefore a monotonic function of (the square root of) the tance.

In a certain way, (the square root of) the tance can be seen as being more fundamental than the distance: it is the simplest algebraic invariant of two non-isotropic points in projective space while the distance involves applying to such algebraic invariant a transcendental function. Unlike the distance, the tance is well-defined for any pair of non-isotropic points. For instance, in view of the above duality, the tance in $\mathcal G$ allows to determine the relative position of the dual hyperplanes (or geodesics, when $n=2$) in $\mathcal K$ and to calculate the corresponding Riemannian quantities (distances and angles between hyperplanes). Similarly, the tance between a point $\pmb p\in\mathcal K$ and a point $\pmb q\in\mathcal G$ allows to calculate the distance between $\pmb p$ and the dual hyperplane $\mathbb P_\mathbb Rq^\perp\cap\mathcal K$. Curiously, the square root of the tance is exactly the Lorentz factor $\gamma$ corresponding to a pair of inertial reference frames $\pmb p,\pmb q$ (see Subsection \ref{sub:tance_lorentz}).

\smallskip

{\bf Isometries.} The restricted Lorentz group $\mathrm{SO}^+(1,n)$ of all linear, orientation and future-preserving isometries of $\mathbb M^{n+1}$ naturally acts on $\mathcal K$ by orientation-preserving isometries ($\mathrm{SO}^+(1,n)$ is in fact isomorphic to the group $\mathrm{PSO}(1,n)$ of orientation-preserving isometries of $\mathcal K$). The non-identical orientation-preserving isometries of $\mathcal K$ can be \textit{elliptic,} \textit{parabolic,} or \textit{hyperbolic.} Consider $n=2$. In this case, the elliptic isometries have exactly one fixed point (its center) in~$\mathcal K$ and, geometrically, they are rotations around the center. The orbit of a point under a one-parameter group generated by an elliptic isometry is a \textit{metric circle,} that is, a locus equidistant from the center. A parabolic isometry has a unique isotropic fixed point $\pmb v$ and the orbit of a point under a one-parameter group generated by such an isometry is a \textit{horocycle,} that is, a curve containing $\pmb v$ that is orthogonal to every geodesic that has $\pmb v$ as a vertex.

Finally, a hyperbolic isometry $I$ has exactly a pair of fixed isotropic points $\pmb v_1,\pmb v_2$. The geodesic $G:=\G{\wr}v_1,v_2{\wr}\cap\mathcal K$ is $I$-stable and, moreover, the orbit of a point under a one-parameter group generated by $I$ is a \textit{hypercycle,} that is, a locus equidistant from $G$. Note that, at the level of Minkowski space, $I$ is what is called a \textit{boost.} Indeed, the geodesic $G$ can be interpreted as a family of inertial observers such that any of these observers sees all the others with velocities in a same direction (see Subsection \ref{sub:rapidityandvelocity}). Now, given inertial observers $\pmb p,\pmb q\in G$ and a hyperbolic isometry stabilizing $G{\wr}\pmb p,\pmb q{\wr}$, the relative velocity between $\pmb p,\pmb q$ and that between $\pmb p,\pmb I(\pmb q)$ (as measured by $\pmb p$) have the same direction.

As we will see, elliptic, hyperbolic, and parabolic isometries play a major role respectively in the Wigner rotation, the relativistic velocities addition, and the Doppler effect.

\smallskip

The above construction endowing (open subspaces of) the projective space with a geometric structure arising from a non-degenerate form on a vector space does not depend on the choice of the signature of the form nor on the field of real numbers. In fact, many other geometries that are relevant in physics can be approached in this manner. This includes Fubini-Study geometries (quantum information theory), anti-de Sitter space (adS/CFT correspondence), and complex hyperbolic geometry (complex Minkowski space). For this reason, in what follows, we will briefly discuss how the above works in more general settings.

\setcounter{equation}{0}

\subsection{Classic geometries}\label{sub:classicgeometries}
Let $V$ be an $(n+1)$-dimensional $\mathbb K$-vector space, where $\mathbb{K}$ is either $\mathbb{R}$ or $\mathbb{C}$ (it is also possible to take a module over the quaternions in place of $V$, see \cite{AGr2011}). We endow $V$ with a nondegenerate symmetric bilinear (respectively, Hermitian) form $\langle -,- \rangle : V \times V \rightarrow \mathbb{K}$ when $\mathbb K=\mathbb R$ (respecitvely, $\mathbb K=\mathbb C$). As in the previous subsection, we will denote by $\pmb p$ a point in projective space $\mathbb P_\mathbb KV$ and by $p\in V\setminus\{0\}$ a representative of $\pmb p$.

The \textit{signature\/} of a point $\pmb p\in\mathbb{P}_{\mathbb{K}}V$ is the sign of $\langle p,p\rangle$ (which can be $-$, $+$, or $0$). The signature is well defined because $\langle kp,kp\rangle=|k|^2\langle p,p\rangle$ for all $0\neq k\in\mathbb K$. It divides $\mathbb P_\mathbb KV$ into \textit{negative,} \textit{positive,} and \textit{isotropic\/} points:
$$\text BV:=\{\pmb p \in \mathbb P_{\mathbb K}V \mid \langle{p,p}\rangle < 0 \}, \qquad\text{E}V:=\{\pmb p \in \mathbb P_{\mathbb K}V \mid \langle{p,p}\rangle > 0 \}, \qquad\text{S}V:=\{\pmb p \in \mathbb P_{\mathbb K}V \mid \langle{p,p}\rangle = 0 \}.$$
The space $\text{S}V$ of isotropic points is called the \textit{absolute.} Note that $\mathcal K$, $\partial\mathcal K$, and $\mathcal G$ in the previous subsection, where $V$ is taken as the Minkowski space $\mathbb M^{n+1}$, correspond respectively to $\mathrm{B}V$, $\mathrm{S}V$, and $\mathrm{E}V$.

Let $\pmb p\in\mathbb{P}_{\mathbb{K}}V \setminus\text{S}V$ be a nonisotropic point. Then
$$V=\mathbb{K}p\oplus p^{\perp},\qquad v=\pi[\pmb p]v+\pi^{\prime}[\pmb p]v$$
where
\begin{equation}\label{eq:projectors}
\pi[\pmb p]:v\mapsto v-\frac{\langle v,p\rangle}{\langle p,p\rangle}p\in p^\perp,\qquad\pi^{\prime}[\pmb p]:v\mapsto\frac{\langle v,p\rangle}{\langle p,p \rangle}p\in \mathbb Kp
\end{equation}
are the orthogonal projectors.

As in \eqref{eq:tangentspace}, we have a natural identification $\mathrm{T}_{\pmb p}\mathbb{P}_{\mathbb{K}}V \simeq \text{Lin}_{\mathbb K}(\mathbb{K}p, p^\perp)$ of the tangent space to $\mathbb{P}_{\mathbb{K}}V$ at a nonisotropic point $\pmb p$ with the space of $\mathbb K$-linear maps from $\mathbb Kp$ to $p^\perp$. Using this identification, we define the pseudo-Riemannian metric
\begin{equation}\label{eq:metric2}
\langle\varphi_1,\varphi_2\rangle_{\pmb p}:=\pm\mathrm{Re}\frac{\big\langle \varphi_1(p),\varphi_2(p)\big\rangle}{\langle p,p \rangle}
\end{equation}
where $\pmb p$ is a nonisotropic point and $\varphi_1,\varphi_2\in\mathrm{T}_{\pmb p}\mathbb P_\mathbb KV$. Clearly, when $\mathbb K=\mathbb C$, this pseudo-Riemannian metric comes from a Hermitian metric (simply do not take the real part in the above expression; the imaginary part of the Hermitian metric is the K\"ahler form).

Let $W$ be a $2$-dimensional real linear subspace $W\subset V$ such that the restriction of the form to $W$ is non-null; in the complex case, we also require the Hermitian form restricted to $W$ to be real. The projectivization $\mathbb P_\mathbb KW$ is called an \textit{extended\/} geodesic (note that, in the complex case, we take the \textit{complex\/} projectivization of the \textit{real\/} subspace $W$). Extended geodesics are always topological circles and their intersections with $\mathrm{B}V$ and $\mathrm{E}V$ provide all the usual geodesics of the corresponding (pseudo-)Riemannian metric connection \cite{AGr2011}.

\begin{example}\label{example:classicgeometries}
Besides the extended (real) hyperbolic space constructed in Subsection \ref{sub:kinematicspace} (a hyperbolic ball glued with de Sitter space along their absolutes), we point out a few other examples:
\begin{itemize}
    \item Let $\mathbb{K}=\mathbb{C}$, let $-++\dots+$ be the signature of the Hermitian form $\langle{-,-}\rangle$, and take the sign $-$ in \eqref{eq:metric2}. In this case, $\text{B}V=:\mathbb H^n_\mathbb C$ is the \textit{complex hyperbolic} space. Complex hyperbolic space is to complex Minkowski space as the real hyperbolic space is to real Minkowski space. Note that, when $\dim_\mathbb CV=2$, both $\mathrm{B}V$ and $\mathrm{E}V$ are Poincar\'e hyperbolic discs isometric to the kinematic space $\mathcal K$ (see Subsection \ref{sub:kinematicspace}).
    \item Let $\mathbb{K}=\mathbb{R}$, let $--+\dots+$ be the signature of the symmetric bilinear form of $\langle{-,-}\rangle$, and take the $-$ sign in \eqref{eq:metric2}. Now, $\text{B}V=:ad\mathbb S^n$ is the \textit{anti-de Sitter\/} space (which appears, say, in relativity and in the adS/CFT correspondence). Note that there is a natural map $ad\mathbb S^{2n+1}\to\mathbb H_\mathbb C^n$, the \textit{anti-Hopf fibration\/}: when $V$ is an $(n+1)$-dimensional complex vector space with a Hermitian form of signature $-+\dots+$, its decomplexification is a $2(n+1)$-dimensional real vector space with a symmetric bilinear form of signature $--+\dots+$ (the real part of the Hermitian form); the fibers of the map $\mathbb P_\mathbb RV\to\mathbb P_\mathbb CV$, $\mathbb Rp\mapsto\mathbb Cp$, are circles. In particular, the fibration $ad\mathbb S^3\to\mathbb H^1_\mathbb C$ can be relevant to special relativity (see the previous item).
    \item Let $\mathbb{K} = \mathbb{C}$, let $+...+$ be the signature of the Hermitian form $\langle{-,-}\rangle$ and take the sign $+$ in \eqref{eq:metric2}. In this case, we obtain the \textit{Fubini-Study\/} metric on the complex projective space $\mathrm{E}V=\mathbb P_\mathbb CV$. The Fubini-Study metric is widely used in the geometry of quantum information (the Bloch sphere corresponds to the case $\dim_\mathbb CV=2$).
    \end{itemize}
\end{example}

Following this approach, it is possible to express many other important (pseudo-)Riemmannian concepts (say, curvature tensor, metric connection, parallel transport) in a similar coordinate-free fashion \cite{AGr2011}. Moreover, all the geometries obtained in this way, including their natural generalization to grassmannians, are Einstein manifolds~\cite{AGr2012}.

\section{The physics of kinematic space}\label{sec:physicsofkinematicspace}
\subsection{Tance and Lorentz factor}\label{sub:tance_lorentz}

Let us first describe the Lorentz factor, the time dilation, and the length contraction at the level of the kinematic space~$\mathcal K$ introduced in Subsection \ref{sub:kinematicspace}.

\smallskip

Let  $\pmb p,\pmb q\in\mathcal K$ be inertial reference frames, and let $p$ be an event that happened at time $t_0=|p|/c$ for $\pmb p$. Hence, $p$ happened at time $t=\big|\pi'[\pmb q]p\big|/c$ for $\pmb q$ and we obtain
\begin{equation*}
\frac t{t_0}=\sqrt{\frac{\Big\langle\pi'[\pmb q]p,\pi'[\pmb q]p\Big\rangle}{\langle p,p\rangle}}=\sqrt{\frac{\Big\langle\frac{\langle p,q\rangle q}{\langle q,q\rangle},\frac{\langle p,q\rangle q}{\langle q,q\rangle}\Big\rangle}{\langle p,p\rangle}}=\sqrt{\ta(p,q)}=:\gamma_{\pmb p,\pmb q},
\end{equation*}
where $\ta(\pmb p,\pmb q)$ is the tance defined in \eqref{eq:tance}. Clearly, $\gamma_{\pmb p,\pmb q}$ is the usual Lorentz factor and $t=\gamma_{\pmb p,\pmb q}t_0$ is nothing but the time dilation (see Proposition \ref{prop:timedilation}).

\begin{rmk}\label{rmk:velocity}
The usual formula for the Lorentz factor in terms of the relative scalar velocity between\/ $\pmb p,\pmb q\in\mathcal K$ can be obtained as follows. Take homogeneous coordinates\/ $[c,x_1,\dots,x_n]$ with\/ $\sum x_i^2\leqslant c^2$ that identify the closed kinematic space with a closed\/ $n$-ball\/ $\overline{\mathbb B}^n$ of radius\/ $c$ centred at\/ $\pmb p=[c,0,0,\dots,0]$. Then, if\/ $\pmb q=[c,v_1,\dots,v_n]$, the relative scalar velocity\/ $v$ between\/ $\pmb p$ and\/ $\pmb q$ is given by the Euclidean distance in $\overline{\mathbb B}^n$ between the observers, that is, $v=\sqrt{\sum v_i^2}$. Hence, we have
\begin{equation*}
\gamma_{\pmb p,\pmb q}=\sqrt{\ta(p,q)}=\sqrt{\frac{c^4}{-c^2(-c^2+\sum v_i^2)}}=\frac1{\sqrt{1-\frac{v^2}{c^2}}}.
\end{equation*}
In particular, in terms of the tance, the relative scalar velocity between $\pmb p,\pmb q$ is given by
\begin{equation}\label{eq:velocity_tance}
    v=c\,\sqrt{1-\frac1{\ta(\pmb p,\pmb q)}}.
\end{equation}
(For a coordinate-free form of this remark, see Subsection \ref{sub:rapidityandvelocity}.)
\end{rmk}

\begin{prop}[ (time dilation)]\label{prop:timedilation}
Let\/ $\pmb p,\pmb q\in\mathcal K$ be inertial observers and let\/ $w\in\mathbb M^{n+1}\setminus\{0\}$ be an event that happened at time\/ $t_{\pmb p}\ne0$ for\/ $\pmb p$ and at time\/ $t_{\pmb q}$ for\/ $\pmb q$. Then
$$\frac{t_{\pmb q}}{t_{\pmb p}}=\sqrt{\frac{\ta(\pmb q,\pmb w)}{\ta(\pmb p,\pmb w)}}.$$
{\rm(}Note that the formula is also well-defined when\/ $w$ is lightlike because the term\/ $\langle w,w\rangle$ cancels out.{\rm)} In particular, when\/ $\pmb w=\pmb p$, we obtain\/ $t_{\pmb q}=\gamma_{\pmb q,\pmb p}t_{\pmb p}$.
\end{prop}
\begin{proof}
Follows directly from $t_{\pmb p}^2=\big\langle\pi'[p]w,\pi'[p]w\big\rangle=\frac{\langle p,w\rangle\langle w,p\rangle}{\langle p,p\rangle}$ and $t_{\pmb q}^2=\big\langle\pi'[q]w,\pi'[q]w\big\rangle=\frac{\langle q,w\rangle\langle w,q\rangle}{\langle q,q\rangle}$.
\end{proof}

Taking $\frac{t_{\pmb p}}{t_{\pmb q}},\frac{t_{\pmb q}}{t_{\pmb p}}$ as projective coordinates, one can think of time dilation as a function $\mathbb M^{n+1}\setminus\{0\}\to\mathbb P_\mathbb R^1$; this allows to accommodate the cases when the event $w$ happens at time $t=0$ for (exactly) one of the inertial reference frames.

\begin{prop}[ (length contraction)]\label{prop:spacecontraction}
Let\/ $\pmb p,\pmb q\in\mathcal K$ be inertial observers and assume that\/ $\pmb p$ observes a rigid rod at rest as having length\/ $\ell_{\pmb p}$. We represent the rod by a spacelike vector\/ $w\in p^\perp\setminus\{0\}$. Then,
$$\frac{\ell_{\pmb q}}{\ell_{\pmb p}}=\sqrt{1+\frac{\ta(\pmb q,\pmb w)}{\ta(\pmb p,\pmb q)}},$$
where\/ $\ell_{\pmb q}$ stands for the length of the rod as measured by\/ $\pmb q$. In particular, if\/ $p,q,w$ are coplanar\/ {\rm(}that is, the rod is in the direction of the relative velocity between\/ $\pmb p$ and\/ $\pmb q${\rm)}, then $\ell_{\pmb p}=\gamma_{\pmb p,\pmb q}\ell_{\pmb q}$.
\end{prop}
\begin{proof}
We have $\ell_{\pmb p}=|w|$ and $\ell_{\pmb q}=|w'|$, where $w':=w-\frac{\langle w,q\rangle}{\langle p,q\rangle}p$ (note that $w'\in q^\perp$ and that $w'$ belongs to the straight line through $w$ parallel to $\mathbb Rp$). Therefore,
$$\ell_{\pmb q}^2=\bigg\langle w-\frac{\langle w,q\rangle}{\langle p,q\rangle}p,w-\frac{\langle w,q\rangle}{\langle p,q\rangle}p\bigg\rangle=\langle w,w\rangle\bigg(1+\frac{\langle w,q\rangle\langle q,w\rangle\langle p,p\rangle}{\langle p,q\rangle\langle q,p\rangle\langle w,w\rangle}\cdot\frac{\langle q,q\rangle}{\langle q,q\rangle}\bigg)=\ell_{\pmb p}^2\bigg(1+\frac{\ta(\pmb q,\pmb w)}{\ta(\pmb p,\pmb q)}\bigg).$$
When $p,q,w$ are coplanar, the determinant
$\det\left[\begin{smallmatrix}
\langle p,p\rangle&\langle p,q\rangle&0\\
\langle q,p\rangle&\langle q,q\rangle&\langle q,w\rangle\\
0&\langle w,q\rangle&\langle w,w\rangle
\end{smallmatrix}\right]$
vanishes, that is, $\ta(w,q)+\ta(p,q)=1$ which implies the result.
\end{proof}

Given $\pmb p,\pmb q\in\overline{\mathcal K}$, the geometric configuration corresponding to the coplanar case in the above proposition is unique. Indeed, $\pmb w$ must be the point orthogonal to $\pmb p$ in the extended geodesic $\G\wr{\pmb p},{\pmb q}\wr$ (because $w\in p^\perp$ and the coplanarity of $p,q,w$ means that $\pmb w$ belongs to $\G{\wr}\pmb p,\pmb q{\wr}$). Similarly, $\pmb w'$ must be the point in $\G\wr{\pmb p},{\pmb q}\wr$ orthogonal to $\pmb q$. Moreover, it is curious to note that the formula $\ell_{\pmb p}=\gamma_{\pmb p,\pmb q}\ell_{\pmb q}$ is actually a direct consequence of the geometric identity $\ta(\pmb p,\pmb q)=\ta(\pmb w,{\pmb w}')$ (whose proof is a straightforward calculation). Indeed, we have
$$\gamma_{\pmb p,\pmb q}^2=\ta(\pmb p,\pmb q)=\ta(\pmb w,\pmb w')=\frac{\langle w,w'\rangle\langle w',w\rangle}{\langle w,w\rangle\langle w',w'\rangle}=\frac{\langle w,w\rangle\langle w,w\rangle}{\langle w,w\rangle\langle w',w'\rangle}=\frac{|w|^2}{|w'|^2}=\frac{\ell_{\pmb p}^2}{\ell_{\pmb q}^2}$$
since $w'=w-\frac{\langle w,q\rangle}{\langle p,q\rangle}p$ and $w\in p^\perp$.

\subsection{Rapidity, velocity, and parallel transport}\label{sub:rapidityandvelocity}

{\bf Rapidity and rapidity addition.} Given an inertial observer $\pmb p\in\mathcal K$, we call the tangent space $\T_{\pmb p}\mathcal K$ the \textit{space of rapidities\/} at $\pmb p$. A tangent vector $w\in\T_{\pmb p}\mathcal K$ is the \textit{relative rapidity,} as measured by $\pmb p$ (or, simply, at $\pmb p$), between $\pmb p$ and the inertial observer $\pmb q:=\exp_{\pmb p}w$, where $\exp$ stands for the Riemannian exponential map. Hence, the hyperbolic distance between $\pmb p,\pmb q$ is $d(\pmb p,\pmb q)=|w|$.

There is a natural way to sum rapidities at $\pmb p\in\mathcal K$ that takes into account the geometry of the kinematic space. After introducing it, we will relate rapidity and velocity in order to show that the geometric sum of rapidities leads to the relativistic velocities addition.

\begin{defi}\label{defi:sum_of_rapidities}
Let $\pmb p\in\mathcal K$ be an inertial observer and let $w_1,w_2\in\T_{\pmb p}\mathcal K$ be rapidities. Take $\pmb q:=\exp_{\pmb p}w_1$, let $w_2'\in\T_{\pmb q}\mathcal K$ be the parallel transport of $w_2$ along the geodesic segment joining $\pmb p$ and $\pmb q$, and let $\pmb r:=\exp_{\pmb q}w_2'$. We define the \textit{sum of rapidities\/} $w_1\oplus w_2\in\T_{\pmb p}\mathcal K$ as the unique rapidity $w\in\T_{\pmb p}\mathcal K$ such that $\exp_{\pmb p}w=\pmb r$. Equivalently, $w_1\oplus w_2:=\exp_{\pmb p}^{-1}\pmb r$ (see Figure 2).
\end{defi}

Clearly, the above definition works in any Riemannian manifold with infinite injectivity radius and, in the particular case of an Euclidean vector space, it coincides with the vector space sum. (In fact, Definition \ref{defi:sum_of_rapidities} can be seen as a straightforward generalization of the vector sum in an Euclidean vector space.)

\smallskip

{\bf Scaled rapidity.} While rapidities live in the tangent spaces to points in the kinematic space, {\it scaled rapidities\/} (a.k.a {\it hyperbolic velocities\/}) appear naturally as tangent vectors to points in the  \textit{scaled kinematic space\/} $\mathcal K^c$. In order to introduce the scaled kinematic space we will use the following remark.

\begin{wrapfigure}{R}{0pt}\label{fig:rapidity_addition}
\includegraphics[width=0.25\linewidth]{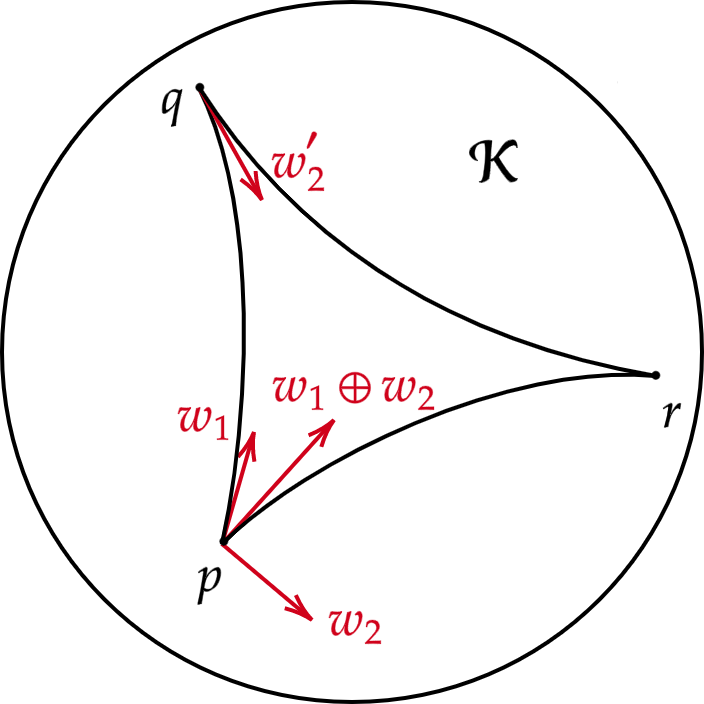}
\caption{Rapidity addition}
\end{wrapfigure}

\begin{rmk}\label{rmk:representative}
Once a representative $p\in\mathbb M^{n+1}$ of $\pmb p\in\mathcal K$ is chosen, we identify $\T_{\pmb p}\mathcal K$ with $p^\perp$ via \eqref{eq:tangentspace}, that is, via the map $\varphi\mapsto\varphi(p)\in p^\perp$, $\varphi\in\T_{\pmb p}\mathcal K\simeq\mathrm{Lin}(\mathbb Rp,p^\perp)$. (There is, however, a \emph{natural\/} identification $\T_{\pmb p}\mathcal K\simeq p^\perp$, see Remark~\ref{rmk:natural_identification}.)
\end{rmk}

The (open) scaled kinematic space is the manifold $\mathcal K$ endowed with a different Riemannian metric as follows. Given $\pmb p\in\mathcal K$, we take the future-directed representative $p\in\mathbb M^{n+1}$ such that $\langle p,p\rangle=-c^2$ and identify $\T_{\pmb p}\mathcal K\simeq p^\perp$ as in Remark~\ref{rmk:representative}. Now, we equip $\T_{\pmb p}\mathcal K$ with the inner product in $p^\perp$ (that is, the restriction of the symmetric bilinear form in $\mathbb M^{n+1}$ to $p^\perp$). Provided with such Riemannian metric, the manifold $\mathcal K$ is called the \textit{scaled kinematic space\/} $\mathcal K^c$. The scaled kinematic space $\mathcal K^c$ is a hyperbolic space of constant curvature $-1/c^2$ because it is isometric to the future sheet of the hyperboloid $\langle x,x\rangle=-c^2$ with the induced metric from Minkowski space. The concepts of {\it space of scaled rapidities,} of {\it relative scaled rapidity,} and of {\it sum of scaled rapidities\/} are analogous to their rapidity counterparts.

Let $w\in\T_{\pmb p}\mathcal K$ be a relative rapidity at $\pmb p$ which correponds to the relative scaled rapidity $w_c\in\T_{\pmb p}\mathcal K^c$. It follows from~\eqref{eq:metric} that $|w_c|=c|w|$, where the left-hand side (respectively, the right-hand side) norm is the one in $\T_{\pmb p}\mathcal K^c$ (respectively, in $\T_{\pmb p}\mathcal K$).

\begin{rmk}\label{rmk:natural_identification}
Let $\pmb p\in\mathcal K$. There is a natural identification $\T_{\pmb p}\mathcal K\simeq p^\perp\subset\mathbb M^{n+1}$ because, given $\varphi\in\T_{\pmb p}\mathcal K=\mathrm{Lin}(\mathbb Rp,p^\perp)$, there exists a unique future-oriented representative $p\in\mathbb M^{n+1}$ such that $u:=\varphi(p)\in p^\perp$ satisfies $\langle u,u\rangle=\langle\varphi,\varphi\rangle_{\pmb p}$. Clearly, $\langle p,p\rangle=-1$. Analogously, there is a natural identification $\T_{\pmb p}\mathcal K^c\simeq p^\perp\subset\mathbb M^{n+1}$ and the corresponding representative of $p$ in this case satisfies $\langle p,p\rangle=-c^2$.
\end{rmk}

{\bf Velocity.} Velocity and (relative) rapidity are concepts of different natures because velocity is algebraic. Let us introduce the space of velocities at a point $\pmb p\in\mathcal K^c$ and endow it with its natural geometric structure.

\smallskip

Given $\pmb p,\pmb q\in\mathcal K^c$, we define the {\it relative velocity\/} between $\pmb p,\pmb q$ at $\pmb p$ as the simplest algebraic expression (in the sense that it does not depend on the choice of representatives) for a tangent vector $v\in\T_{\pmb p}\mathcal K^c$ that is tangent to the geodesic $\G{\wr}\pmb p,\pmb q{\wr}$ at $\pmb p$:

\begin{defi}\label{defi:relative_velocity}
Given $\pmb p\in\mathcal K^c$, the relative velocity $v\in\T_{\pmb p}\mathcal K^c\simeq\mathrm{Lin}(\mathbb Rp,p^\perp)$ between $\pmb p$ and $\pmb q\in\overline{\mathcal K}^c$ at $\pmb p$ is defined as the linear map $v=\langle-,p\rangle\displaystyle\frac{\pi[\pmb p]q}{\langle q,p\rangle}$, where $\langle-,p\rangle$ stands for the linear functional $x\mapsto\langle x,p\rangle$, $x\in\mathbb M^{n+1}$.
\end{defi}

By \cite[Lemma 5.2]{AGr2011}, the relative velocity between $\pmb p$ and $\pmb q$ at $\pmb p$ is tangent to the geodesic $G{\wr}\pmb p,\pmb q{\wr}$. So, the relative (scaled) rapidity and the corresponding relative velocity between inertial observers $\pmb p,\pmb q$ at $\pmb p$ have the same direction.

\begin{prop}\label{prop:velocity_is_velocity} Under the identification\/ $\T_{\pmb p}\mathcal K^c\simeq p^\perp$ in Remark\/ \ref{rmk:natural_identification}, the above definition of relative velocity coincides with the usual one.
\end{prop}
\begin{proof}\renewcommand{\qedsymbol}{}
Let $\pmb p\in\mathcal K^c$ and let $\pmb q\in\overline{\mathcal K}^c$. At the level of Minkowski space, the usual relative velocity between $\mathbb Rp,\mathbb Rq$ as measured by $\mathbb Rp$ has the norm given in equation~\eqref{eq:velocity_tance} and the direction of the projection $\pi[\pmb p]q\in p^\perp$ for a future-oriented $q$. On the other hand, the tangent vector $\langle-,p\rangle\frac{\pi[\pmb p]q}{\langle q,p\rangle}$ corresponds, via the identication $\T_{\pmb p}\mathcal K^c\simeq p^\perp$, to $-c^2\frac{\pi[\pmb p]q}{\langle q,p\rangle}$. It remains to observe that $\langle p,q\rangle<0$ (since both are future-oriented) and that
\begin{equation*}
\bigg\langle-c^2\frac{\pi[\pmb p]q}{\langle q,p\rangle},-c^2\frac{\pi[\pmb p]q}{\langle q,p\rangle}\bigg\rangle=c^2\bigg(1-\frac1{\ta(\pmb p,\pmb q)}\bigg).\tag*{$\square$}
\end{equation*}
\end{proof}

The symmetric bilinear form restricted to $\mathbb Rp+\mathbb Rq$, where $\pmb p,\pmb q$ are as in the proof above, has signature $-+$. Hence, the determinant of the Gram matrix
$\left[\begin{smallmatrix}\langle p,p\rangle&\langle p,q\rangle\\
\langle q,p\rangle&\langle q,q\rangle\end{smallmatrix}\right]$ is negative which implies that $\ta(\pmb p,\pmb q)\geqslant1$. The norm of the velocity in Definition \ref{defi:relative_velocity} is therefore always less or equal than $c$. So, the relative velocities at $\pmb p$ constitute the closed $n$-ball $\mathcal V_{\pmb p}\subset\T_{\pmb p}\mathcal K^c$ of radius $c$ centered at $0\in\T_{\pmb p}\mathcal K^c$. Such closed ball is called the \textit{space of velocities\/} $\mathcal V_p$ at $\pmb p$. (This definition can be seen as a coordinate-free form of Remark \ref{rmk:velocity}.)

\medskip

{\bf Hyperbolic structure on $\mathcal V_{\pmb p}$}.
Besides the inner product inherited from $p^\perp$, the space of velocities $\mathcal V_{\pmb p}$ has a natural hyperbolic structure induced from $\mathcal K^c$: we simply send a velocity $v\in\mathcal V_{\pmb p}$ to the inertial observer $\pmb q\in\mathcal K^c$ such that the relative velocity between $\pmb p,\pmb q$ at $\pmb p$ equals $v$ and equip $\mathcal V_{\pmb p}$ with the pullback metric. From the perspective of Minkowski space (see Figure 3), this is nothing but (1) associating a vector $v\in p^\perp$ satisfying $\langle v,v\rangle<c^2$ to the inertial observer $\mathbb R(p+v)$, where $p$ is the future-oriented representative of $\pmb p$ with $\langle p,p\rangle=-c^2$, and (2) equipping $p+\mathbb B^n\simeq\mathcal V_{\pmb p}$ with the hyperbolic metric that comes from the stereographic projection onto the hyperboloid $\langle x,x\rangle=-c^2$, where $\mathbb B^n\subset p^\perp$ stands for the open ball of radius $c$ centred at the origin. Note that, while rapidity is intended to measure the distance between inertial reference frames, the role of scaled rapidity is to measure the ``distance between velocities'' in a velocity space $\mathcal V_{\pmb p}$.

\begin{figure}[h!]
\centering
\includegraphics[width=0.25\textwidth]{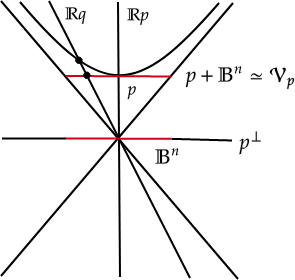}
\caption{Hyperbolic structure on $\mathcal V_{\pmb p}$ (at the level of Minkowski space)}
\label{fig:hyperboloid}
\end{figure}

{\bf Relativistic velocity addition.}
The relative velocity $v\in\mathcal V_{\pmb p}$ between $\pmb p,\pmb q\in\mathcal K^c$ at $\pmb p$ and the corresponding relative scaled rapidity $w_c\in\T_{\pmb p}\mathcal K^c$ are related by
\begin{equation}\label{eq:velocity_scaled_rapidity}
v=v(w_c)=c\Big(\tanh\big(|w_c|/c\big)\Big)\frac{w_c}{|w_c|}
\end{equation}
because those tangent vectors have the same direction and
$$|v|^2=c^2\bigg(1-\frac{1}{\ta(\pmb p,\pmb q)}\bigg)=c^2\bigg(1-\frac{1}{\cosh^2\big(d^c(\pmb p,\pmb q)/c\big)}\bigg)=c^2\tanh^2{\frac{d^c(\pmb p,\pmb q)}c}=c^2\tanh^2\frac{|w_c|}c$$
by Remark \ref{rmk:velocity}, where $d^c(\pmb p,\pmb q)$ stands for the distance function in $\mathcal K^c$. In particular, $v=v(w)=c\big(\tanh|w|\big)\frac w{|w|}$, where $w\in\T_{\pmb p}\mathcal K$ stands for the rapidity between $\pmb p,\pmb q$ at $\pmb p$.

\begin{defi}\label{defi:sum_of_velocities}
Let $v_1,v_2\in\mathcal V_{\pmb p}$ be velocities and let $w_1,w_2\in\T_{\pmb p}\mathcal K^c$ be the corresponding scaled rapidities. We define $v_1\oplus v_2$ simply as the velocity that corresponds to $w_1\oplus w_2$, that is, $v_1\oplus v_2:=v(w_1\oplus w_2)$. (One can also take rapidities instead of scaled rapidities here.)
\end{defi}

\begin{prop} The above definition of velocity addition coincides with the usual relativistic velocity addition.
\end{prop}
\begin{proof}
Let $\pmb p,\pmb q,\pmb r\in\mathcal K^c$ be inertial observers, let $v_1\in\T_{\pmb p}\mathcal K^c$ be the relative velocity between $\pmb p,\pmb q$ at $\pmb p$, and let $v_2'\in\T_{\pmb q}\mathcal K^c$ be the relative velocity between $\pmb q,\pmb r$ at $\pmb q$. The parallel transport $v_2\in\T_{\pmb p}\mathcal K^c$ of $v_2'$ along the geodesic segment joining $\pmb q$ and $\pmb p$ can be interpreted as the relative velocity between $\pmb q,\pmb r$ as measured by $\pmb p$. Indeed, let $I$ be the hyperbolic isometry that stabilizes $G{\wr}\pmb p,\pmb q{\wr}$ and satisfies $I(\pmb q)=\pmb p$. It is easy to see that $I_*(w_c')=w_c$, where $I_*$ stands for the differential of $I$ and $w_c$, $w_c'$ denote respectively the scaled rapidities corresponding to $v_2$, $v_2'$. By the naturality of the exponential map (see \cite[Proposition 5.20]{leemanifolds}, for instance), $\exp_{\pmb p}w_c=\exp_{\pmb p}\big(I_*(w_c')\big)=I(\exp_{\pmb q}(w_c'))=I(r)$. At the level of Minkowski space, the boost $\widetilde I$ corresponding to $I$ sends the pair of inertial observers $\mathbb Rq,\mathbb Rr$ to $\mathbb Rp,\mathbb R\widetilde I(r)$ and the relative velocity between the last two observers, as measured by $\mathbb Rp$, is therefore exactly the relative velocity between the first two ones as measured by $\mathbb Rp$.
\end{proof}

{\bf ``Parallelogram'' law.} Let us take a closer look at the geometry of the sum of velocities. Given velocities $v_1,v_2\in\mathcal V_{\pmb p}$, where $\pmb p\in\mathcal K^c$,
we can assume that $\mathcal V_{\pmb p}$ is an open disk in the two-dimensional subspace of $\T_{\pmb p}\mathcal K^c$ generated by $v_1,v_2$. Now, the sum $v_1\oplus v_2$ is obtained simply by applying to $v_2$ the hyperbolic isometry $I$ (in the sense of the hyperbolic structure of $\mathcal V_{\pmb p}$) that sends the null vector $0$ to $v_1$ and stabilizes the geodesic $G:=\G{\wr}0,v_1{\wr}$. Note that the sum of velocities is noncommutative because, if we apply to $v_1$ the hyperbolic isometry $I'$ that sends $0$ to $v_2$ and stabilizes the geodesic $\G{\wr}0,v_2{\wr}$ then, in general, $I(v_2)\ne I'(v_1)$. In other words, at a first glance, it seems that there is no ``parallelogram'' law for the relativistic addition of velocities. However, this is the case only if we require the parallelogram to be \textit{geodesic\/}; substituting one of the sides for a hypercycle, that is, for a curve that is equidistant from a geodesic, there is indeed a ``parallelogram law'' where the ``parallelogram'' has vertices $0,v_1,v_1\oplus v_2,v_2$ and the sides are the geodesic segment joining $0,v_1$, the geodesic segment joining $v_1,v_1\oplus v_2$, the segment of the hypercycle $H$ of $G$ joining $v_1\oplus v_2,v_2$, and the geodesic segment joining $v_2,0$. In other words, $v_1\oplus v_2$ is obtained by the geometric construction that follows. Draw: the geodesic $G$ joining $0,v_1$; the geodesic $G'$ joining $0,v_2$; the geodesic $G''$ through $v_1$ such that the oriented angle from $G$ to $G''$ at $v_1$ equals that from $G$ to $G'$ at $0$; the hypercycle $H$ of $G$ through $v_2$. Then, $v_1\oplus v_2$ is given by the intersection $H\cap G''$.

\begin{rmk}
This construction of the relativistic velocity addition can also be seen as a geometric realization of the M\"obius addition discussed by A.~Ungar; this follows from the above considerations and from the fact that Poincar\'e's hyperbolic disk $\mathbb H^1_\mathbb C$ (see Example \ref{example:classicgeometries}) is isometric to $\mathcal K$ when $\dim\mathcal K=2$. More precisely, given $\pmb o,\pmb p,\pmb q\in\mathbb H^1_\mathbb C$, we define $\pmb p\oplus_{\pmb o}\pmb q:=I(\pmb q)$, where $I$ stands for the hyperbolic isometry that stabilizes the geodesic $G{\wr}\pmb o,\pmb p{\wr}$ and satisfies $I(\pmb o)=\pmb p$. This is a coordinate-free geometric form of the M\"obius addition formula in \cite[Section 3.4]{ungar2001}: take the unitary disk $\mathbb D$ in $\mathbb C$ centered at the origin (which plays the role of $\pmb o$) and define $a\oplus_Mb:=(a+b)/(1+\overline{a}b)$ for all $a,b\in\mathbb D$.

Similarly, one can give a geometric description of the M\"obius subtraction $a\ominus_Mb:=a\oplus_M(-b)$ by defining $-\pmb q:=R(\pmb o)\pmb q$ and $\pmb p\ominus_{\pmb o}\pmb q:=R(\pmb m)\pmb q$, where $R(\pmb o)$ and $R(\pmb m)$ stand respectively for the reflection in $\pmb o$ and in the middle point $\pmb m$ of the geodesic segment joining $\pmb o$ and $\pmb p$. Indeed, the hyperbolic isometry $I$ that stabilizes the geodesic $G{\wr}\pmb o,\pmb p{\wr}$ and satisfies $I(\pmb o)=\pmb p$ can be written as $I=R(\pmb m)R(\pmb o)$. Now, $\pmb p\oplus_{\pmb o}(-\pmb q)=I(-\pmb q)=I\big(R(\pmb o)\pmb q\big)=R(\pmb m)R(\pmb o)R(\pmb o)\pmb q=R(\pmb m)\pmb q$.
\end{rmk}

Another geometric way to look at the relativistic velocities addition is the following. In order to obtain $v_1\oplus v_2$, we first project $v_2$ orthogonally (in the hyperbolic sense) over the direction of $v_1$ thus obtaining the \textit{horizontal component\/} $v$ of $v_2$. Now, if $v$ and $v_1$ have the same direction, we add $v_1$ and $v$ by simply taking the velocity $v_1\oplus v=v\oplus v_1\in\mathcal V_{\pmb p}$ that lies in the geodesic $G:=\mathrm G{\wr}0,v_1{\wr}$ and satisfies $d^c(0,v_1\oplus v)=d^c(0,v_1)+d^c(0,v)$, where $d^c$ stands for the hyperbolic distance in $\mathcal V_{\pmb p}$ (the case when $v$ and $v_1$ have opposite directions is handled similarly). Finally, it~remains to take the unique velocity $v_1\oplus v_2\in\mathcal V_{\pmb p}$ that is on the same side of $G$ as $v_2$, whose orthogonal projection onto $G$ is $v_1\oplus v$, and whose distance to $G$ equals that of $v_2$ (in other words, the \textit{vertical component\/} of $v_1\oplus v_2$ is the same as that of $v_2$).

\begin{figure}[htp]
\centering
\includegraphics[width=.7\textwidth]{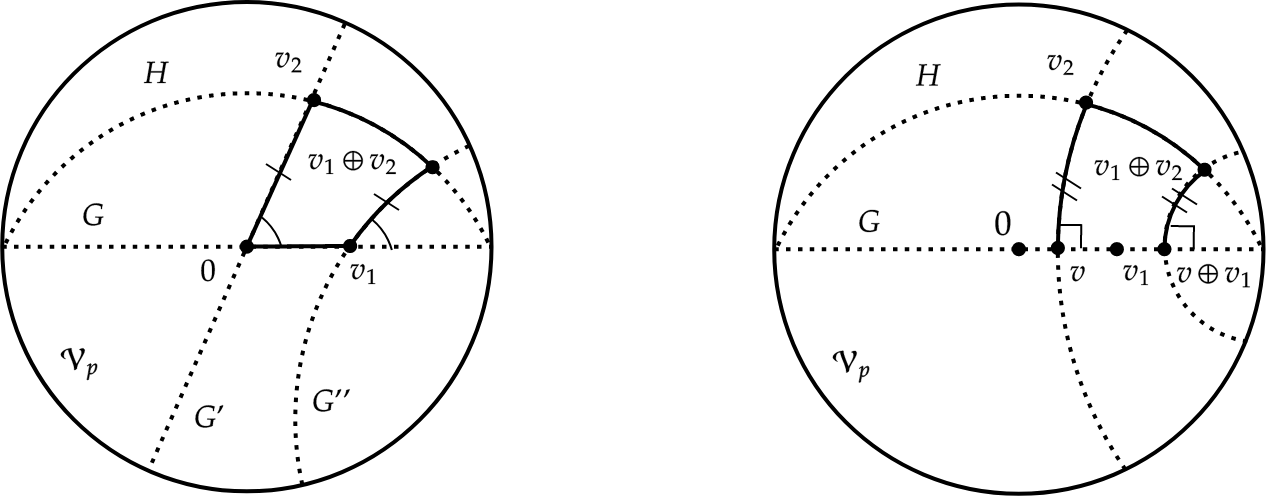}
\caption{``Parallelogram'' law and component sum}
\label{fig:hyperboloid}
\end{figure}

\subsection{Relativistic Doppler effect}\label{sub:doppler}

The relativistic Doppler effect can also be seen in a geometric way.\footnote{We thank J.~A.~Hoyos for suggesting that horocycles should be related to the relativistic Doppler effect.} In this section, we can assume (without loss of generality) that $\dim\mathcal K=2$.

A metric circle $C$ in $\mathcal K$ is the locus of inertial observers that see a given inertial observer $\pmb q\in\mathcal K$ (the center of the circle) with a same given energy. Indeed, $C=\{\pmb p\in\mathcal K\mid\ta(\pmb p,\pmb q)=r\}$, $r>0$, and the energy of $\pmb q$ as measured by $\pmb p$ is determined by $\gamma_{\pmb p,\pmb q}=\sqrt{\ta(\pmb p,\pmb q)}$. In the limit where $\pmb q$ goes to the absolute (and $r$ is fixed) this metric circle turns into a horocycle tangent to the absolute at a point $\pmb f\in\partial\mathcal K$ and the energy being measured by the inertial observers corresponding to points in this horocycle becomes that of the photon $\pmb f$. In other words, the function that assigns to each inertial observer in $\mathcal K$ the energy (or, equivalently, the frequency) that it measures for the photon $\pmb f$ is constant along horocycles (in fact, horocycles will be the level curves of this function, see Corollary \ref{cor:level_curves}). Let us formalize this argument.

\begin{lemma}\label{lemma:half_level_curves}
Let\/ $\pmb f\in\partial\mathcal K$ and let\/ $\pmb r,\pmb r'\in\mathcal K$ be inertial observers in a same horocycle containing $\pmb f$. Then, $\nu_{\pmb r}=\nu_{\pmb r'}$, where\/ $\nu_{\pmb r},\nu_{\pmb r'}$ stand for the frequencies of\/ $\pmb f$ as measured respectively by\/ $\pmb r,\pmb r'$.
\end{lemma}
\begin{proof} Let $I: \overline{\mathcal{K}} \rightarrow \overline{\mathcal{K}}$ be the parabolic isometry that fixes $\pmb f$ and maps $\pmb r'$ to $\pmb r$. It is well-known that the energy $E_{\pmb r}$ of the photon $\pmb f$ as measured by $\pmb r$ is given by the magnitude of the projection of the $(n+1)$-momentum of the photon in the direction of $\mathbb Rr$ divided by $c$. Similarly, one can express the energy $E_{\pmb r'}$ of the photon $\pmb f$ as measured by $\pmb r'$, which leads to
$$\bigg(\frac{E_{\pmb r'}}{E_{\pmb r}} \bigg)^2=\frac{\big\langle \pi^{\prime}[\pmb r'] f,\pi^{\prime}[\pmb r']f\big\rangle}{\big\langle\pi^{\prime}[\pmb r] f, \pi^{\prime}[\pmb r] f \big \rangle}=\frac{\langle f,r'\rangle^2\langle r,r\rangle}{\langle f,r\rangle^2\langle r',r'\rangle}=\frac{\big\langle f,r'\big\rangle^2\big\langle\widetilde I(r'),\widetilde I(r')\big\rangle}{\big\langle f,\widetilde I(r')\big\rangle^2\langle r',r'\rangle}=\frac{\big\langle f,r'\big\rangle^2\big\langle\widetilde I(r'),\widetilde I(r')\big\rangle}{\big\langle\widetilde I(f),\widetilde I(r')\big\rangle^2\langle r',r'\rangle}=1,$$
where $\widetilde I$ stands for the element in $\mathrm{SO}^+(1,2)$ corresponding to $I$; it satisfies $\widetilde I(f)=f$ because $I$ is parabolic (see, for instance, \cite{charette}).
\end{proof}

\begin{wrapfigure}{r}{0pt}
\includegraphics[width=0.26\textwidth]{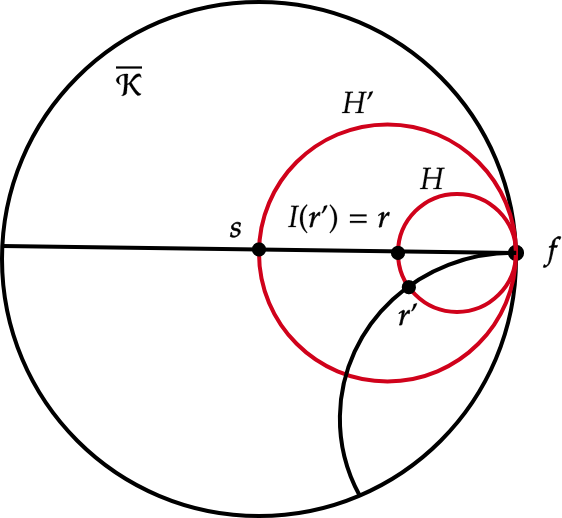}
\caption{Horocycles and the relativistic Doppler effect}
\end{wrapfigure}

Now, consider the case of two inertial observers $\pmb r,\pmb s \in \mathcal K$ which are respectively considered as the receiver and the source of a photon $\pmb f\in\partial\mathcal K$ such that $\pmb r,\pmb s,\pmb f$ are in a same geodesic $G$. Assume that the inertial observers are moving away from each other (it is easy to see that, in order to reach the receiver, the photon that has to be sent by the source is such that $\pmb r$ is in the geodesic segment joining $\pmb s$ and $\pmb f$). Let $\nu_{\pmb s}$ (respectively, $\nu_{\pmb r}$) be the frequency of $\pmb f$ as measured by $\pmb s$ (respectively, by~$\pmb r$). Then (see, for example, \cite[Section 4.3]{rindler})
$$\frac{\nu_{\pmb s}}{\nu_{\pmb r}} = \sqrt{\frac{1 + v/c}{1 - v/c}} = \sqrt{\frac{1 + (e^w - e^{-w})/(e^w + e^{-w})}{1 - (e^w - e^{-w})/(e^w + e^{-w})}}
=e^{w}=e^{d(\pmb r,\pmb s)},$$
\noindent where $v$ and $w$ are respectively the scalar relative velocity and relative rapidity between $\pmb r$ and $\pmb s$. When the inertial observers are moving towards each other (in
this case, the photon $\pmb f'$ to be sent corresponds to the other vertex of $G$) we have
$\nu_{\pmb s}/\nu_{\pmb r}=e^{-d(\pmb r,\pmb s)}$. We are now able to prove the following proposition (for the definition of Busemann function see, for instance, \cite[Section 1.2]{quint2006}).

\begin{prop}[ (relativistic Doppler effect)]\label{prop:Doppler}
Let\/ $\pmb p,\pmb q\in\mathcal K$ be inertial observers and let\/ $\pmb f\in\partial\mathcal K$ be a photon. Let $\nu_{\pmb p}$ and $\nu_{\pmb q}$ be respectively the frequencies of $\pmb f$ as measured by $\pmb p$ and $\pmb q$. We have
$$\frac{\nu_{\pmb p}}{\nu_{\pmb q}}=e^{b_{\pmb f}(\pmb p,\pmb q)},$$ where $b_{\pmb f}$ stands for the Busemann function determined by $\pmb f$.
\end{prop}

\begin{proof}
By Lemma \ref{lemma:half_level_curves}, the ratio $\nu_{\pmb p}/\nu_{\pmb q}$ can be obtained in terms of the distance between the horocycles $H,H'$ containing~$\pmb f$ and passing respectively through $\pmb p,\pmb q$. Now the proof follows from the case of collinear $\pmb p,\pmb q,\pmb f$ which was already considered above.
\end{proof}

A direct consequence of Lemma \ref{lemma:half_level_curves} and Proposition \ref{prop:Doppler} is the following Corollary.

\begin{cor}\label{cor:level_curves} Let\/ $\pmb f\in\partial\mathcal K$ and let\/ $\pmb p,\pmb q\in\mathcal K$ be inertial observers. Then\/ $\nu_{\pmb p}=\nu_{\pmb q}$ if and only if\/ $\pmb p,\pmb q$ belong to a same horocycle containing $\pmb f$, where\/ $\nu_{\pmb p},\nu_{\pmb q}$ stand for the frequencies of\/ $\pmb f$ as measured respectively by\/ $\pmb p,\pmb q$.
\end{cor}

\subsection{Wigner rotation}

Let $\pmb p,\pmb q,\pmb r \in \mathcal{K}$ be inertial observers. A well-known fact in special relativity is that the composition of boosts $\mathbb Rp\to\mathbb Rq\to\mathbb Rr\to\mathbb Rp$ is a spatial rotation called the {\it Wigner rotation.} Let us give a coordinate-free proof of this phenomenon at the level of the kinematic space $\mathcal K$. In the next proposition we consider, without loss of generality, that $\dim\mathcal K=2$ and that the kinematic space is (arbitrarily) oriented.

\begin{prop}[ (Wigner Rotation)]\label{prop:wigner} Let $\pmb p_i \in \mathcal K$, $i=1,2,3$, be inertial observers and let $G_{ij}:=G{\wr}\pmb p_i,\pmb p_j{\wr}$ be the geodesic connecting $\pmb p_i$ and $\pmb p_j$. Let $h_1,h_2,H: \mathcal K \rightarrow \mathcal K$ stand for the hyperbolic isometries such that $h_1$ stabilizes $G_{12}$ and $h_1(\pmb p_1)=\pmb p_2$; $h_2$ stabilizes $G_{23}$ and $h_2(\pmb p_2)=\pmb p_3$; $H$ stabilizes $G_{13}$ and $H(\pmb p_1)=\pmb p_3$. Then $h_2 h_1 = e_\theta H$, where $e_\theta : \mathcal K \rightarrow \mathcal K$ is the elliptic isometry that fixes $\pmb p_3$ and whose angle of rotation $\theta \in [-\pi,\pi]$ is minus the oriented area of the triangle with vertices $(\pmb p_1,\pmb p_2,\pmb p_3)$.
\end{prop}

\begin{proof} Let $\pmb q_1\in G_{12}$ be the middle point of the geodesic segment joining $\pmb p_1,\pmb p_2$ and let $\pmb q_2\in G_{23}$ the middle point of the geodesic segment joining $\pmb p_2,\pmb p_3$. We have $h_1 = R_2 R_1$ where $R_1$ stands for the reflection in the geodesic orthogonal to $G_{12}$ passing through $\pmb q_1$ and $R_2$, for the reflection in the geodesic orthogonal to $G_{12}$ passing through $\pmb p_2$. Similarly, $h_2 = R_4 R_3$ where $R_3$ denotes the reflection in the geodesic orthogonal to $G_{23}$ passing through $\pmb p_2$ and $R_4$, the reflection in the geodesic orthogonal to $G_{23}$ passing through $\pmb q_2$. Lastly, let $R_5$ and $R_6$ be the reflections in the geodesics orthogonal to $G{\wr}\pmb q_1,\pmb q_2{\wr}$ passing respectively through $\pmb q_2$ and $\pmb q_1$ and let $h_3:\mathcal K \rightarrow \mathcal K$, $h_3 := R_5 R_6$, be a hyperbolic isometry that stabilizes the geodesic $G{\wr}\pmb q_1,\pmb q_2{\wr}$.

Note that $R_1 R_6$, $R_3 R_2$, and $R_5 R_4$ are elliptic isometries such that $R_1R_6=\sigma_2\sigma_1$, $R_3 R_2=\sigma_3\sigma_2$, and $R_5 R_4=\sigma_1\sigma_3$, where $\sigma_1,\sigma_2,\sigma_3$ stand respectively for the reflections in the geodesics $G{\wr}\pmb q_1,\pmb q_2{\wr}$, $G{\wr}\pmb q_1,\pmb p_2{\wr}$, and $G{\wr}\pmb p_2,\pmb q_2{\wr}$. Hence,
$$R_5 h_2 h_1 R_6 = (R_5 R_4) (R_3 R_2) (R_1 R_6) = (\sigma_1\sigma_3)(\sigma_3\sigma_2)(\sigma_2\sigma_1)=1$$
which implies $h_2 h_1 = R_5 R_6 = h_3$. Now, note that $h_3H^{-1}(\pmb p_3)=h_3(\pmb p_1)=h_2\big(h_1(\pmb p_1)\big)=\pmb p_3$, and $h_3$ is obviously not the inverse of $H$, so $h_3 H^{-1}$ has to be an elliptic isometry $e_{\theta}$ fixing $\pmb p_3$. In other words, $h_2h_1=e_{\theta}H$.

\begin{figure}[htp]
\centering
\includegraphics[width=.6\textwidth]{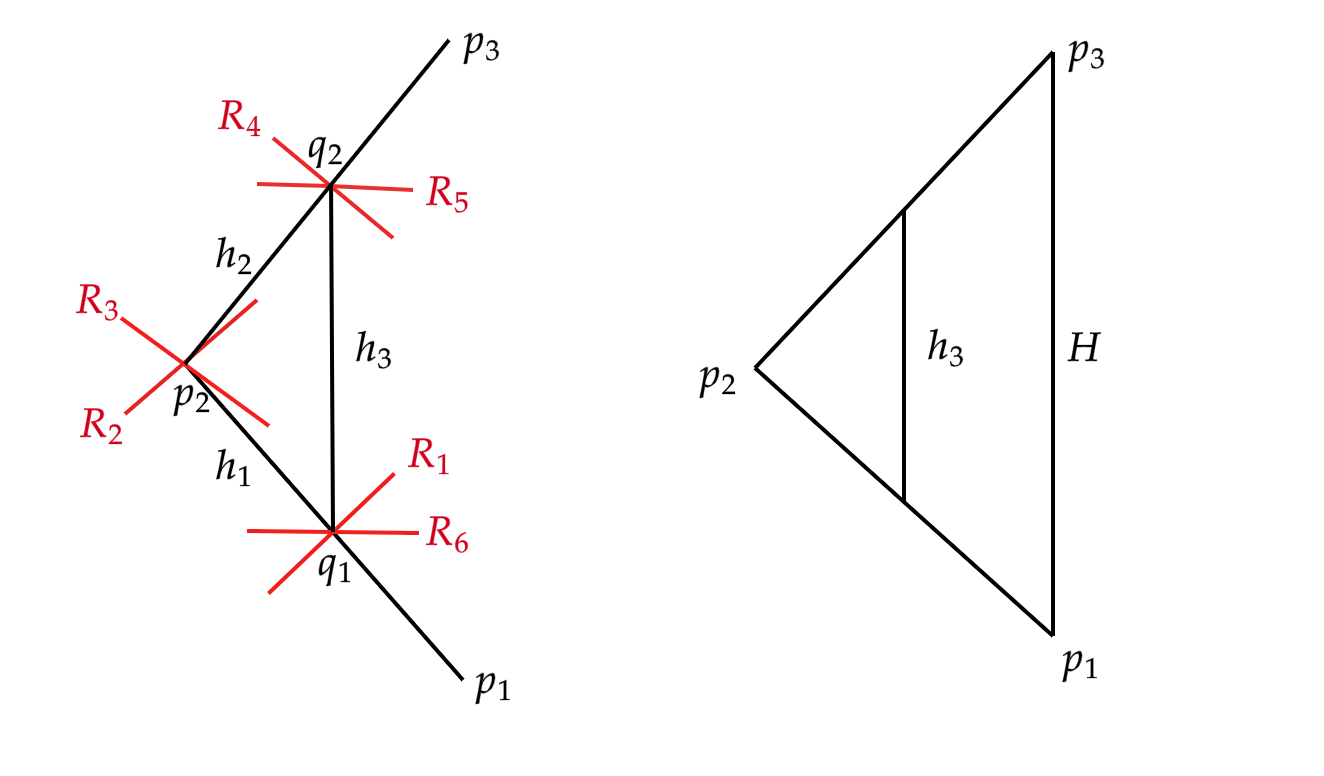}
\caption{Proof of Proposition \ref{prop:wigner}}
\label{fig:wignerproof}
\end{figure}

The differential of a hyperbolic isometry, being applied to a vector tangent at a point of its stable geodesic, coincides with the parallel transport along this geodesic. So, since $h_2 h_1 H^{-1} = e_\theta$, we conclude that $\theta$ is minus the oriented area of the triangle $(\pmb p_1,\pmb p_2,\pmb p_3)$ (the minus sign comes from the fact that the sum of the internal angles of a geodesic triangle in $\mathcal K$ is less than $\pi$ or, equivalently, from the Gauss-Bonnet theorem).
\end{proof}

\begin{rmk}
Wigner rotation can also be seen as a measure of the non-commutativity of the rapidity addition (see Definition \ref{defi:sum_of_rapidities}) as follows. Let $w_1,w_2\in\mathrm{T}_{\pmb p_1}\mathcal K$ be rapidities at $\pmb p_1\in\mathcal K$. Moreover, define $\pmb p_2:=\mathrm{exp}_{\pmb p_1}w_1$, $\pmb p_3 := \text{exp}_{\pmb p_1} (w_1 \oplus w_2)$, $\pmb q_2 := \text{exp}_{\pmb p_1} w_2$, and $\pmb q_3 := \text{exp}_{\pmb p_1} (w_2 \oplus w_1)$. The triangles $(\pmb p_1,\pmb p_2,\pmb p_3)$ and $(\pmb q_1,\pmb q_2,\pmb q_3)$ are clearly congruent and it is straightforward to see that the angle $\theta$ at $\pmb p_1$ between the geodesic ray joining $\pmb p_1,\pmb p_3$ and the geodesic ray joining $\pmb p_1,\pmb q_3$ is given by $\theta=\pi - \sum_i \alpha_i=\mathrm{Area}(\pmb p_1,\pmb p_2,\pmb p_3)$, where the $\alpha_i$'s stand for the internal angles of the triangle $(\pmb p_1,\pmb p_2,\pmb p_3)$.
\end{rmk}

\subsection{An invariant of three points and causality}\label{sub:causality}

Let us take a look at a relativistic interpretation of the algebraic invariant
\begin{equation}\label{eq:invariant_tau}
\eta(\pmb p,\pmb q,\pmb u):=\frac{\langle u,p\rangle\langle p,q\rangle\langle q,u\rangle}{\langle p,p\rangle\langle q,q\rangle\langle u,u\rangle}
\end{equation}
of two inertial observers $\pmb p,\pmb q\in\mathcal K$ and a point $\pmb u\in\mathcal G$ in de Sitter space.

The invariant $\eta(\pmb p,\pmb q,\pmb u)$ determines whether $\pmb p$ and $\pmb q$ agree or disagree on the order of occurrence of an event that happened at time $t=0$ and a space-like event $u\in\mathbb M^{n+1}$. Indeed, the observers agree or disagree respectively when the sign of
$$\frac{\big\langle\pi'[\pmb p]u,\pi'[\pmb q]u\big\rangle}{\langle u,u\rangle}=\frac{\Big\langle\frac{\langle u,p\rangle}{\langle p,p\rangle}p,
\frac{\langle u,q\rangle}{\langle q,q\rangle}q\Big\rangle}{\langle u,u\rangle}=
\eta(\pmb p,\pmb q,\pmb u)$$
is negative or positive. At the level of the extended kinematic space, this can be translated as follows: the observers agree/disagree exactly when $\pmb p,\pmb q$ lie in the same/in distinct components of $\mathcal K\setminus G$, where $G$ is the geodesic with polar point~$u$ (this can be inferred by looking at the relative position between $\mathbb Rp$, $\mathbb Rq$, and $u^\perp$). A usual way of saying that there will always exist observers that do not agree on the occurrence order of spacelike separated events is that causality is not well defined for this kind of events.

\subsection{Dynamics}\label{sub:dynamics}

At a first glance it may seem that, when passing from Minkowski space to kinematic space, one loses information, obtaining a space that models well kinematic phenomena but is not suited to described dynamics. This subsection is intended to illustrate that this is not the case.

Let $\xi:I\to\mathbb M^{n+1}$ be a smooth curve such that $\xi(0)=0$, $\big\langle\dot\xi(\tau),\dot\xi(\tau)\big\rangle=-c^2$ (that is, $\xi$ is parameterized by proper time), and $\dot\xi(\tau)$ is future-oriented for every $\tau\in I$. It gives rise to the curve $\zeta(\tau)=\mathbb P_\mathbb R\dot\xi(\tau)$ in the scaled kinematic space $\mathcal K^c$, where $\mathbb P_\mathbb R\dot\xi(\tau)$ stands for the image of $\dot\xi(\tau)$ under the canonical projection $\mathbb M^{n+1}\to\mathbb P^n_\mathbb R$. Conversely, given a smooth curve $\zeta:I\to\mathcal K^c$, there exists a unique lift $\zeta_0:I\to\mathbb M^{n+1}$ of $\zeta$ to $\mathbb M^{n+1}$ such that $\big\langle\zeta_0(\tau),\zeta_0(\tau)\big\rangle=-c^2$ and $\zeta_0(\tau)$ is future-oriented for every $\tau\in I$. Now, there exists a unique smooth curve $\xi:I\to\mathbb M^{n+1}$ such that $\xi(0)=0$ and $\dot\xi(\tau)=\zeta_0(\tau)$ for every $\tau\in I$.

Let us see that a tangent vector to the curve $\zeta$ is nothing but the $(n+1)$-acceleration of $\xi$ in view of the identification $\T_{\zeta(\tau)}\mathcal K^c\simeq\zeta(\tau)^\perp$ (see Remark \ref{rmk:natural_identification}). On one hand, as a linear map $\T_{\zeta(\tau)}\mathcal K^c=\mathrm{Lin\big(\mathbb R\zeta(\tau),\zeta(\tau)^\perp\big)}$,
$$\dot\zeta(\tau):\zeta_0(\tau)\mapsto\pi\big[\zeta(\tau)\big]\dot\zeta_0(\tau)=\pi\big[\zeta(\tau)\big]\Ddot{\xi}(\tau)$$
by \cite[Lemma A.1]{AGG2011}. On the other hand, $\pi\big[\zeta(\tau)\big]\Ddot{\xi}(\tau)=\Ddot{\xi}(\tau)$ since $\big\langle\dot\xi(\tau),\dot\xi(\tau)\big\rangle$ is constant.

The curve $\zeta$ can be interpreted as the list of inertial frames occupied by the observer with worldline $\xi$ (that is, $\zeta(\tau)$ is the inertial frame occupied at the instant $\tau$). Note that, if $\zeta$ is constant, $\zeta(\tau)=\pmb p$ for every $\tau$, then $\xi$ is a straight line in $\mathbb M^{n+1}$ passing through the origin (the worldline $\mathbb Rp$ of an inertial observer, as expected); when $\zeta$ is a geodesic, $\xi$ is a hyperbola that represents a motion with constant $(n+1)$-acceleration (a.k.a. hyperbolic motion).

Finally, let $A=A(\pmb p,\tau)$, $\pmb p\in\mathcal K$, $\tau\in\mathbb R$, be a smooth time-dependent vector field in $\mathcal K$. Let $\zeta$ be the maximal integral curve of $A$ corresponding to the initial conditions $\pmb p_0\in\mathcal K$ and $\tau_0\in\mathbb R$, that is, $\dot\zeta(\tau)=A\big(\zeta(\tau),\tau\big)$ and $\zeta(\tau_0)=\pmb p_0$ (such an integral curve exists and is unique by \cite[Theorem 9.48]{leemanifolds}). The $\xi$ obtained from $\zeta$ as above is nothing but the dynamics associated to the time-dependent force field $F=mA$, where $m$ is the rest mass of an observer whose worldline is $\xi$.

\bibliographystyle{plain}
\bibliography{references}

\bigskip

\noindent
{\sc Rafael Ferreira}

\noindent
{\sc Departamento de Matem\'atica, ICMC, Universidade de S\~ao Paulo, S\~ao Carlos, Brazil}

\noindent
rafael.ferreira.pereira@usp.br

\medskip

\noindent
{\sc João dos Reis Junior}

\noindent
{\sc Departamento de Matem\'atica, ICMC, Universidade de S\~ao Paulo, S\~ao Carlos, Brazil}

\noindent
joao.reis.reis@usp.br

\medskip

\noindent
{\sc Carlos H.~Grossi}

\noindent
{\sc Departamento de Matem\'atica, ICMC, Universidade de S\~ao Paulo, S\~ao Carlos, Brazil}

\noindent
grossi@icmc.usp.br

\end{document}